\documentclass[a4paper,12pt]{amsart}
\usepackage[mathscr]{euscript}
\usepackage{amsfonts}
\usepackage{amssymb, amsmath}
\usepackage{array,booktabs,mathtools,tabularx}
\DeclarePairedDelimiter\abs\lvert\rvert
\newcolumntype{L}{@{}X@{}}
\usepackage{url}
\usepackage{fixme}
 \newtheorem{thm}{Theorem}
  \newtheorem{cor}[thm]{Corollary}
   \newtheorem{defn}[thm]{Definition}
\newcommand{\F}{\mathbb F_q}
\begin{document}
\title{Forms and Linear Network Codes}
\author{Johan~P.~Hansen}
\address{Department of Mathematics, Aarhus University}
\email{matjph@imf.au.dk}
\date{\today}
\maketitle
\begin{abstract}

We present a general theory to obtain linear network codes utilizing  forms and obtain explicit families of equidimensional vector spaces, in which any pair of distinct vector spaces intersect in the same small dimension.
The theory is inspired by  the methods  of the author utilizing the osculating spaces of Veronese varieties.

Linear network coding transmits information in terms of a basis of a vector space and the information is received as a basis of a possibly altered vector space. Ralf Koetter and Frank R. Kschischang 
introduced a metric on the set af vector spaces and showed that a minimal distance decoder for this metric achieves correct decoding if the dimension of the intersection of  the transmitted and received vector space is sufficiently large.

The vector spaces in our construction are equidistant in the above metric and the distance between any pair of vector spaces is large making them suitable for linear network coding.

The parameters of the resulting linear network codes are determined.

\end{abstract}


\subsection*{Notation}
\begin{itemize}
\item $\F$ is the finite field with $q$ elements of characteristic $p$.
\item $\mathbb F=\overline{\mathbb F_q}$ is an algebraic closure of $\F$.
\item $R_d = \mathbb F [X_0,\dots,X_n]_d$ and $R_d(\F)= \F  [X_0,\dots,X_n]_d$ the homogenous polynomials  of degree $d$ with coefficients in $\mathbb F$ and $\F$.
\item $R=\mathbb F [X_0,\dots,X_n] = \oplus_d  R_d$ and $R(\F)=\F [X_0,\dots,X_n] = \oplus_d  R_d(\F)$.
\item $G(l,N)$ is the Grassmannian of $l$-dimensional $\mathbb F$-linear subspaces of $\mathbb F^N$ and $G(l,N)(\F)$ its $\F$-rational points, i.e. $l$-dimensional $\F$-linear subspaces of $\F^N$.

\end{itemize}
\section{Introduction}
\subsection{Linear network coding}\label{network}
In linear network coding transmission is obtained by transmitting a number of packets into the network and each packet is regarded as a vector of length $N$ over a finite field $\F$. The packets travel the network through intermediate nodes, each forwarding $\F$-linear combinations of the packets it has available. Eventually the receiver tries to infer the originally transmitted packages from the packets that are recieved, see \cite{DBLP:citeseer_10.1.1.11.697} and \cite{Ho06arandom}.

Ralf Koetter and Frank R. Kschischang \cite{DBLP:journals/tit/KoetterK08}  endowed the Grassmannian $G(l,N)(\F)$ of $l$-dimensional $\F$-linear subspaces of $\F^N$ with the metric
\begin{equation}\label{dist}
\mathrm{dist}(V_1,V_2):=\dim_{\F}(V_1+V_2)-\dim_{\F}(V_1\cap V_2),
\end{equation}
where $V_1,V _2\in G(l,N)(\F)$.
\begin{defn}
A linear network code $\mathcal C \subseteq G(l,N)(\F)$ is a set of $l$-dimensional $\F$-linear subspaces of $\F^N$.

The size of the code $\mathcal C \subseteq G(l,N)(\F)$  is denoted by $\vert \mathcal C \vert$ and the minimal distance by
\begin{equation}
D(\mathcal C):= \min_{V_1,V_2 \in \mathcal C, V_1 \neq V_2} \mathrm{dist}(V_1,V_2)\ .
\end{equation}
The linear network code $\mathcal C$ is said to be of type $[N,l,\log_q \vert \mathcal C \vert, D(\mathcal C)]$. Its normalized weight is $\lambda = \frac{l}{N}$, its rate is $R= \frac{\log_q (\vert \mathcal C \vert)}{N l}$
and its normalized minimal distance  is $\delta = \frac{D(\mathcal C)}{2 l}$.
\end{defn}
Ralf Koetter and Frank R. Kschischang showed that a minimal distance decoder for this metric achieves correct decoding if the dimension of the intersection of  the transmitted and received vector-space is sufficiently large. Also they obtained  Hamming, Gilbert-Varshamov and Singleton coding bounds.

\section{Explicit construction of linear network codes from normalized homogenous polynomials}
Let $\mathcal N(e)$ be the normalized homogenous polynomials over $\F$ of degree $e$ in $X_0,\dots,X_n$ corresponding to monic polynomials in the single variable case.
Specifically, $\mathcal N (e)= \F  [X_0,\dots,X_n]_e /\sim$, where $F_1 \sim F_2$ iff $F_1=c F_2$ for some  constant $c \in \F^*$.  Let $\mathcal I(e)\subseteq \mathcal N(e)$ be the irreducible normalized homogenous polynomials. 

For any subset $\mathcal B \subseteq \F  [X_0,\dots,X_n]_e/\sim $ of normalized homogenous polynomials of degree $e$, we define the linear network code $\mathcal C_{\mathcal B}$ as a  collection of $\F$-linear subspaces $V_G$, one for each $\ G \in \mathcal B$, in the vector space of all homogenous forms of degree $d$ .

\begin{defn}\label{deff}
Let $G \in \F  [X_0,\dots,X_n]_e/\sim$ be a normalized homogenous polynomial in $X_0,\dots,X_n$ of degree $e$ with coefficients in $\F$. Assume that $G\neq 0$.

Let $d \geq e$ and consider the $\F$-linear injective morphism
\begin{equation}\label{morphism}
\begin{split}
\F  [X_0,\dots,X_n]_{d-e}\quad \rightarrow&\quad  \F  [X_0,\dots,X_n]_d\\
F \quad \mapsto& \quad G \cdot F
\end{split}
\end{equation}
with image \begin{equation}\label{VG}
V_G:=G \cdot \F  [X_0,\dots,X_n]_{d-e}\subseteq \F  [X_0,\dots,X_n]_d = \F^N\ ,
\end{equation}
which is a $\F$-linear subspace of dimension $l=\binom{n+d-e}{n}$ in the ambient vector space of dimension $N=\binom{n+d}{n}$.

For any subset $\mathcal B \subseteq \F  [X_0,\dots,X_n]_e/\sim $ of normalized homogenous polynomials of degree $e$, the linear network code $\mathcal C_{\mathcal B}\subseteq G(l,N)(\F)$ consists of all the linear subspaces in the vector space $\F  [X_0,\dots,X_n]_d$ of homogenous forms of degree $d$ with $d \geq e$, that are realized as images (\ref{VG}) for some $G \in \mathcal B$.
\begin{equation}
\mathcal C_{\mathcal B} = \{ V_G = G \cdot \F  [X_0,\dots,X_n]_{d-e} \vert \ G \in \mathcal B\} \subseteq G(l,N)(\F)\ .
\end{equation}
\end{defn}

In \cite{2012arXiv1210.7961H} we studied the resulting linear network codes $\mathcal C_B$, when $\mathcal B$ is the set of normalized homogenous polynomials  which are powers of linear terms, see Corollary \ref{linear}.
This amounted to the study of the osculating spaces of Veronese varieties.

In the present paper we present the resulting linear network codes $\mathcal C_B \subseteq G(l,N)(\F)$, when $\mathcal B$ more generally is any set of normalized
homogenous polynomials in $\F  [X_0,\dots,X_n]_e/\sim$, where each pair of unequal polynomials has the constants as their only common divisors generalising the above result. In particular we treat the case where $\mathcal B$ is the set of all irreducible normalized polynomials, see Corollary \ref{irr}.

\begin{thm}
Let $\mathcal B \subset \F  [X_0,\dots,X_n]_e/\sim$ be a set of normalized homogenous polynomials of degree $e$, such that
for any pair $G_1, G_2 \in \mathcal B$,  with $G_1 \neq G_2$, their only commen divisors are the constants.

 For $d \geq e$, let  $\mathcal C_{\mathcal B} \subseteq G(l,N)(\F)$  be the corresponding linear network code, as defined in Definition \ref{deff}.

The packet length of $\mathcal C_{\mathcal B}$ is $N= \binom{n+d}{n}$, the dimension of the ambient vector space. The vector spaces in the linear network code are equidimensional of dimension $l=\binom{n+d-e}{n}$ as $\F$-linear subspaces of the ambient  $N=\binom{n+d}{n}$-dimensional $\F$-vectorspace $\mathbb F_q^N$.

The number of vector spaces in the linear network code $\mathcal C_{\mathcal B}$  is the number $\vert \mathcal B \vert$ of elements in $\mathcal B$.

The elements in the code are equidistant in the metric $\mathrm{dist} (V_1,V_2)$ of (\ref{dist}) of Section \ref{network}. Let $V_1, V_2 \in \mathcal C_{\mathcal B}$ be vector spaces with $V_1 \neq V_2$.

If $d-e < e$, then $\dim_{\F}(V_1 \cap V_2) =0$ and 
\begin{equation}\label{zero}
\mathrm{dist} (V_1,V_2)=2 \ \binom{n+d-e}{n}
\end{equation}
If $d-e \geq e$, then $\dim_{\F}(V_1 \cap V_2) =\binom{n+d-2e}{n}$ and
\begin{equation}\label{nonzero}
\mathrm{dist} (V_1,V_2)=2\  \Bigg(\binom{n+d-e}{n}-\binom{n+d-2e}{n}\Bigg)
\end{equation}
\end{thm}
\begin{proof}
The morphism in (\ref{morphism}) maps injectively to the $\F$-vectorspace $\F  [X_0,\dots,X_n]_d$ of dimension $N= \dim_{\F}\F  [X_0,\dots,X_n]_d = \binom{n+d}{n}$. The dimension of the image is
$l=\dim_{\F} V_G=  \dim_{\F}\F  [X_0,\dots,X_n]_{d-e} = \binom{n+d-e}{n}$.

Let $V_{G_1}, V_{G_2} \in \mathcal C_{\mathcal B}$ be two vector spaces with $V_{G_1} \neq V_{G_2}$. Assume that  $H \in V_{G_1} \cap V_{G_2}$ with $H \neq 0$. Then $H=G_1 F_1 =G_2 F_2$ with $F_i \in \F  [X_0,\dots,X_n]_{d-e}$. By unique factorization this implies that $G_1$ divides $F_2$ and there is a unique $K \in \F  [X_0,\dots,X_n]_{d -2e}$ such that $F_2=G_1 K$.

If $\deg F_2 =d-e<e =\deg G_1$, this is impossible and $V_{G_1} \cap V_{G_2} = 0$ proving (\ref{zero}) by the definition of the metric in (\ref{dist}).

If $\deg F_2 =d-e\geq e =\deg G_1$ then $H= G_2 F_2 = G_2 G_1 K$ for a unigue  $K \in \F  [X_0,\dots,X_n]_{d-2e}$. Therefore 
\begin{equation}\dim _{\F} V_{G_1} \cap V_{G_2} = \dim_{\F}\F  [X_0,\dots,X_n]_{d-2e} =  \binom{n+d-2e}{n}\end{equation} proving (\ref{nonzero}) by the definition of the metric in (\ref{dist}).

\end{proof}
\subsection{The case when $\mathcal B$ is the set of all irreducible normalized polynomials}

Let $\mathcal N (e) =\F  [X_0,\dots,X_n]_e /\sim$ be the normalized homogenous polynomials of degree $e$.

As $\F  [X_0,\dots,X_n]_e $ is a vector space af dimension $\binom{n+e}{n}$ over $\F$, we have the following formula for the number $N(e)$ of normalized  homogenous polynomials of degree $e$:
\begin{equation}
N(e):=\vert \mathcal N(e) \vert =\frac{q^{\binom{n+e}{n}}-1}{q-1}\ .
\end{equation}

Let $\mathcal I(e)\subseteq \mathcal N(e)$ be the irreducible normalized homogenous polynomials and let $I(e):=\vert \mathcal I(e)\vert $ be their number. 

The number af products of $a_i$ elements from $\mathcal I(i)$ is $\binom{\mathcal I(i)+a_i-1}{a_i}$. Using unique factorisation in $\F  [X_0,\dots,X_n]$, we get that
\begin{equation}
N(e)= \sum_{1 a_1+ 2 a_2+\dots e a_e =e}\binom{ I(1)+a_1-1}{a_1}\dots \binom{ I(i)+a_e-1}{a_e}\ .
\end{equation}
and the rekursive formula for $I(e)$:
\begin{equation}\label{rekurs}
\begin{split}
I(e)=& N(e)\\
-& \sum_{1 a_1+ 2 a_2+\dots (e-1) a_{e-1} =e}\binom{ I(1)+a_1-1}{a_1}\dots \binom{ I(i)+a_{e-1}-1}{a_{e-1}}
\end{split}
\end{equation}
Similar expressions in the non-homogenous case are obtained  in \cite{MR2444942}, \cite{MR2516427},
\cite{MR0172872} and \cite{MR0153665}.

The same rekursive method permits to construct the elements in $\mathcal I(e)$ in the present case.

\begin{cor}\label{irr}
Let $\mathcal I \subset \F  [X_0,\dots,X_n]_e/\sim$ be the set of all irreducible normalized homogenous polynomials of degree $e$.

 For $d \geq e$, let  $\mathcal C_{\mathcal I} \subseteq G(l,N)(\F)$  be the corresponding linear network code, as defined in Definition \ref{deff}.

The packet length of $\mathcal C_{\mathcal I}$ is $N= \binom{n+d}{n}$, the dimension of the ambient vector space. The vector spaces in the linear network code are equidimensional of dimension $l=\binom{n+d-e}{n}$ as linear subspaces of the ambient  $N=\binom{n+d}{n}$-dimensional $\F$-vectorspace $\F^N$.

The number of vector spaces in the linear network code $\mathcal C_{\mathcal I}$  is the number $\vert \mathcal I \vert$ of elements in $\mathcal I$ and is determined rekursively by the formula (\ref{rekurs}).

The elements in the code are equidistant in the metric $\mathrm{dist} (V_1,V_2)$ of (\ref{dist}) of Section \ref{network}. Let $V_1, V_2 \in \mathcal C_{\mathcal B}$ be vector spaces with $V_1 \neq V_2$.

If $d-e < e$, then $\dim_{\F}(V_1 \cap V_2) =0$ and 
\begin{equation}
\mathrm{dist} (V_1,V_2)=2 \ \binom{n+d-e}{n}
\end{equation}
If $d-e \geq e$, then $\dim_{\F}(V_1 \cap V_2) =\binom{n+d-2e}{n}$ and
\begin{equation}
\mathrm{dist} (V_1,V_2)=2\  \Bigg(\binom{n+d-e}{n}-\binom{n+d-2e}{n}\Bigg)\ .
\end{equation}
\end{cor}
Parameters for the linear network codes $\mathcal C_{\mathcal I(e)} \subseteq G(l,N)(\mathbb F_2)$ constructed from $\mathbb F_2[X_0,X_1,X_2]$ are given in Table \ref{tabel} for $d=1,2,\dots, 10$ and $e=1,2,\dots,5$.

\subsection{The case when $\mathcal B$ is the set of powers of linear normalized polynomials}

In \cite{2012arXiv1210.7961H} we studied the resulting linear network codes $\mathcal C_B$, when $\mathcal B$ is the set of normalized homogenous polynomials  which are powers of linear terms.
This amounted to the study of the osculating spaces of Veronese varieties.

Let $\mathcal L(e)\subseteq \mathcal N(e)$ be the set of $e$-fold powers of normalized homogenous linear polynomials and let $\vert \mathcal L(e)\vert = N(1)= \frac{q^{\binom{n+1}{n}}-1}{q-1}$ denote their number.

\begin{cor}\label{linear}

For $d \geq e$, let  $\mathcal C_{\mathcal L} \subseteq G(l,N)(\F)$  be the corresponding linear network code, as defined in Definition \ref{deff}.

The packet length of $\mathcal C_{\mathcal L}$ is $N= \binom{n+d}{n}$, the dimension of the ambient vector space. The vector spaces in the linear network code are equidimensional of dimension $l=\binom{n+d-e}{n}$ as linear subspaces of the ambient  $N=\binom{n+d}{n}$-dimensional $\F$-vectorspace. 

The number of vector spaces in the linear network code $\mathcal C_{\mathcal L}$  is the number 
$\vert \mathcal L(e)\vert = N(1)= \frac{q^{\binom{n+1}{n}}-1}{q-1}$ of elements in $\mathcal L$.
The elements in the code are equidistant in the metric $\mathrm{dist} (V_1,V_2)$ of (\ref{dist}) of Section \ref{network}. Let $V_1, V_2 \in \mathcal C_{\mathcal B}$ be vector spaces with $V_1 \neq V_2$.

If $d-e < e$, then $\dim_{\F}(V_1 \cap V_2) =0$ and 
\begin{equation}
\mathrm{dist} (V_1,V_2)=2 \ \binom{n+d-e}{n}
\end{equation}
If $d-e \geq e$, then $\dim_{\F}(V_1 \cap V_2) =\binom{n+d-2e}{n}$ and
\begin{equation}
\mathrm{dist} (V_1,V_2)=2\  \Bigg(\binom{n+d-e}{n}-\binom{n+d-2e}{n}\Bigg)\ .
\end{equation}
\end{cor}
\begin{table}\caption{Parameters for the linear network codes $\mathcal C_{\mathcal I(e)} \subseteq G(l,N)(\mathbb F_2)$ with $\mathcal I(e) \subseteq \mathbb F_2[X_0,X_1,X_2]_e$, see Corollary \ref{irr}}\label{tabel}
  \centering
  \tiny 
\begin{tabularx}{\textwidth}{
    c L r L c  
    c 
    r L r L r L r L r L r L r L r L r L r}
\toprule
    &&                     && $d$       && 1     && 2     && 3     && 4     && 5     && 6     && 7     && 8     && 9     && 10    \\
$e$ && $\abs{\mathcal C}$  && $N$       && 3     && 6     && 10    && 15    && 21    && 28    && 36    && 45    && 55    && 66    \\
\midrule
$1$ && $7$                 && $l$       && 1     && 3     && 6     && 10    && 15    && 21    && 28    && 36    && 45    && 55    \\
    &&                     && $D$       && 2     && 6     && 6     && 8     && 10    && 12    && 14    && 16    && 18    && 20    \\
    &&                     && $\lambda$ && 0,333 && 0,500 && 0,600 && 0,667 && 0,714 && 0,750 && 0,778 && 0,800 && 0,818 && 0,833 \\
    &&                     && $\delta$  && 1,000 && 1,000 && 0,500 && 0,400 && 0,333 && 0,286 && 0,250 && 0,222 && 0,200 && 0,182 \\
    &&                     && $R$       && 0,936 && 0,156 && 0,047 && 0,019 && 0,009 && 0,005 && 0,003 && 0,002 && 0,001 && 0,001 \\
\midrule
$2$ && $35$                && $l$       &&       && 1     && 3     && 6     && 10    && 15    && 21    && 28    && 36    && 45    \\
    &&                     && $D$       &&       && 2     && 6     && 12    && 14    && 18    && 22    && 26    && 30    && 34    \\
    &&                     && $\lambda$ &&       && 0,167 && 0,300 && 0,400 && 0,476 && 0,536 && 0,583 && 0,622 && 0,655 && 0,682 \\
    &&                     && $\delta$  &&       && 1,000 && 1,000 && 1,000 && 0,700 && 0,600 && 0,524 && 0,464 && 0,417 && 0,378 \\
    &&                     && $R$       &&       && 0,855 && 0,171 && 0,057 && 0,024 && 0,012 && 0,007 && 0,004 && 0,003 && 0,002 \\
\midrule
$3$ && $694$               && $l$       &&       &&       && 1     && 3     && 6     && 10    && 15    && 21    && 28    && 36    \\
    &&                     && $D$       &&       &&       && 2     && 6     && 12    && 20    && 24    && 30    && 36    && 42    \\
    &&                     && $\lambda$ &&       &&       && 0,100 && 0,200 && 0,286 && 0,357 && 0,417 && 0,467 && 0,509 && 0,545 \\
    &&                     && $\delta$  &&       &&       && 1,000 && 1,000 && 1,000 && 1,000 && 0,800 && 0,714 && 0,643 && 0,583 \\
    &&                     && $R$       &&       &&       && 0,944 && 0,210 && 0,075 && 0,034 && 0,017 && 0,010 && 0,006 && 0,004 \\
\midrule
$4$ && $26089$             && $l$       &&       &&       &&       && 1     && 3     && 6     && 10    && 15    && 21    && 28    \\
    &&                     && $D$       &&       &&       &&       && 2     && 6     && 12    && 20    && 30    && 36    && 44    \\
    &&                     && $\lambda$ &&       &&       &&       && 0,067 && 0,143 && 0,214 && 0,278 && 0,333 && 0,382 && 0,424 \\
    &&                     && $\delta$  &&       &&       &&       && 1,000 && 1,000 && 1,000 && 1,000 && 1,000 && 0,857 && 0,786 \\
    &&                     && $R$       &&       &&       &&       && 0,978 && 0,233 && 0,087 && 0,041 && 0,022 && 0,013 && 0,008 \\
\midrule
$5$ && $1862994$           && $l$       &&       &&       &&       &&       && 1     && 3     && 6     && 10    && 15    && 21    \\
    &&                     && $D$       &&       &&       &&       &&       && 2     && 6     && 12    && 20    && 30    && 42    \\
    &&                     && $\lambda$ &&       &&       &&       &&       && 0,048 && 0,107 && 0,167 && 0,222 && 0,273 && 0,318 \\
    &&                     && $\delta$  &&       &&       &&       &&       && 1,000 && 1,000 && 1,000 && 1,000 && 1,000 && 1,000 \\
    &&                     && $R$       &&       &&       &&       &&       && 0,992 && 0,248 && 0,096 && 0,046 && 0,025 && 0,015 \\
\bottomrule
\end{tabularx}
\end{table}

\newcommand{\etalchar}[1]{$^{#1}$}

\end{document}